\documentclass[12pt]{article}

\usepackage{a4wide, amsmath,amsthm,amsfonts,amscd,amssymb,eucal,bbm,mathrsfs}

\def\RR{{\mathbb R}}
\def\CC{{\mathbb C}}

\def\ZZ{{\mathbb Z}}

\def\A{{\mathcal A}}

\def\D{{\mathcal D}}
\def\E{{\mathcal E}}

\def\H{{\mathcal H}}

\def\M{{\mathcal M}}
\def\N{{\mathcal N}}
\def\O{{\mathcal O}}
\def\P{{\mathcal P}}
\def\R{{\mathcal R}}

\def\T{{\mathcal T}}
\def\U{{\mathcal U}}

\def\d{\delta}
\def\e{\varepsilon}
\def\f{\varphi}

\def\i{\iota}

\def\L{{\mathrm L}}

\def\R{{\mathrm R}}
\def\s{\sigma}

\def\t{\tau}

\def\Ad{{\hbox{\rm Ad}}\,}

\def\id{{\rm id}}
\def\1{{\mathbbm 1}}

\def\axb{{\mathbf{P}}}

\def\u1{U(1)}

\def\conf{{\rm Conf(\E)}}
\def\diff{{\rm Diff}}

\def\diffr{\diff(\RR)}
\def\diffs1{\diff(S^1)}

\def\supp{{\rm supp}}
\def\interior{{\rm int}}
\def\slim{{{\mathrm{s}\textrm{-}\lim}}}
\def\wlim{{{\mathrm{w}\textrm{-}\lim}}}

\def\tin{{\mathrm{in}}}
\def\tout{{\mathrm{out}}}

\def\timesi{{\overset{\tin}\times}}
\def\timeso{{\overset{\tout}\times}}
\def\psl2r{{\rm PSL}(2,\RR)}
\def\2dmob{{\overline{\psl2r}\times\overline{\psl2r}}}
\def\<{\langle}
\def\>{\rangle}

\def\poincare{{\P^\uparrow_+}}
\def\regions{\O}

\newtheorem{theorem}{Theorem}[section]
\newtheorem{definition}[theorem]{Definition}
\newtheorem{corollary}[theorem]{Corollary}
\newtheorem{proposition}[theorem]{Proposition}
\newtheorem{lemma}[theorem]{Lemma}
\theoremstyle{remark}
\newtheorem{remark}[theorem]{Remark}

\title{Noninteraction of waves in two-dimensional conformal field theory}
\date{}
\author{
{\bf Yoh Tanimoto\footnote{Supported in part by the ERC Advanced Grant 227458
OACFT ``Operator Algebras and Conformal Field Theory''.}}\\
Dipartimento di Matematica, Universit\`a di Roma ``Tor
Vergata''\\ Via della Ricerca Scientifica, 1 - I--00133 Roma, Italy.\\
E-mail: {\tt tanimoto@mat.uniroma2.it}}

\begin{document}
\maketitle
\begin{abstract}
In higher dimensional quantum field theory, irreducible representations
of the Poincar\'e group are associated with particles.
Their counterpart in two-dimensional massless models 
are ``waves'' introduced by Buchholz. In this paper we show that waves do not interact
in two-dimensional M\"obius covariant theories and in- and out-asymptotic fields coincide.
We identify the set of the collision states of waves with the subspace generated
by the chiral components of the M\"obius covariant net from the vacuum.
It is also shown that Bisognano-Wichmann property, dilation covariance and
asymptotic completeness (with respect to waves) imply M\"obius symmetry.

Under natural assumptions, we observe that the maps which give asymptotic fields in Poincar\'e
covariant theory are conditional expectations between appropriate algebras.
We show that a two-dimensional massless theory is
asymptotically complete and noninteracting if and only if it is a chiral M\"obius covariant theory.
\end{abstract}

\section{Introduction}\label{introduction}

Quantum field theory (QFT) is designed to describe interactions between elementary particles
and can successfully account for a wide range of physical phenomena. 
However, its mathematical foundations are still unsettled and  constitute an active
area of research in mathematical physics. While the most important open problem
in QFT is the existence of interacting models in physical four-dimensional
spacetime, theories in lower dimensional spacetime have also attracted considerable interest. 
For instance, two-dimensional conformal field theories (CFT), whose infinite dimensional
symmetry group is a powerful tool for structural analysis, have been thoroughly investigated.
Superselection  structure of such theories has been clarified and deep classification results
have been obtained \cite{KL02, KL03}. On the other hand, particle aspects of two-dimensional CFT have 
only recently attracted attention \cite{DT10-1, DT10-2}, although a general framework for
scattering of massless excitations (``waves'') in two-dimensional theories dates back to \cite{Buchholz}.  
In particular, it was shown in \cite{DT10-1} that ``waves'' do not interact in \emph{chiral} conformal field theories. 
In this paper we generalize this result to any CFT\footnote{The terms ``conformal'' and ``M\"obius covariant''
will be clarified in Section \ref{preliminaries}.},
relying on ideas from \cite{BF}. Moreover, we
show that a conformal field theory is asymptotically complete (i.e. collision states of waves span 
the entire Hilbert space) if and and only if it is chiral. This latter result is obtained by a 
careful analysis of the chiral components \cite{Rehren} of the theory.

In view of a large body of highly non-trivial results concerning two-dimensional CFT, both
on the sides of physics and mathematics \cite{DMS, FG, Fuchs},
our assertion that these theories have trivial scattering theory may seem surprising.
In this connection we emphasise that the presence of interaction
in scattering theory \emph{cannot} be inferred solely
 from the fact that a particular expression for the Hamiltonian or the
correlation functions differ from those familiar from free field theory.   
In fact, the Ising model, the most fundamental ``interacting'' model, can be
considered as a subtheory of ``free'' fermionic field \cite{MS}, hence the conventional
term of ``interaction'' seems ill-defined.
Instead, a conclusive argument should rely on a scattering theory which implements, 
in the theoretical setting, the quantum mechanical procedure of state preparation 
at asymptotic times. Such an intrinsic scattering theory was developed by Buchholz \cite{Buchholz} in
the framework of algebraic quantum field theory \cite{Haag}, which we also adopt in this work. 
The elementary excitations of this collision theory, called ``waves'' in \cite{Buchholz},  are eigenstates 
of the relativistic mass operator. However, they are not necessarily particles in the conventional
sense of Wigner i.e. states in an irreducible representation space of the Poincar\'e group. This 
less restrictive concept of the particle is natural in two-dimensional massless theories, where irreducible
representations of the Poincar\'e group typically have infinite multiplicity (cf. Section \ref{representation-theory}).

As the classical results on the absence of interaction in dilation-covariant theories in physical spacetime
require the existence of irreducible representations of the Poincar\'e group with finite multiplicity 
\cite{Dell-Antonio, BF}, they cannot be applied to two-dimensional CFT directly. We combine essential ideas from \cite{BF} with the representation theory of the M\"obius group  to overcome this difficulty and obtain triviality of the scattering matrix. 
Under asymptotic completeness with respect to waves, one can even prove that dilation covariance implies
M\"obius covariance, hence also noninteraction.
Exploiting again the  
M\"obius symmetry, we construct chiral observables following Rehren \cite{Rehren} which live on the positive or negative
lightrays and show that they generate all the collision states of waves from the
vacuum. In examples of non-chiral two-dimensional CFT, the profile of chiral observables is
well-known \cite{KL03}, hence this result gives an explicit description of the subspace of collision states.
As a by-product we obtain an alternative proof of the noninteraction of waves and the insight that asymptotic
completeness with respect to waves of a conformal field theory (in the sense of waves) is equivalent to chirality.
This suggest that chiral M\"obius covariant theories are generic examples of
noninteracting massless theories in two-dimensional spacetime (cf. \cite{Tanimoto11-3}).
Indeed, it turns out that Poincar\'e covariant theory satisfying
Bisognano-Wichmann property and Haag duality is noninteracting and asymptotically complete
(with respect to waves) if and only if it is chiral.
This is a strengthened converse of the sufficient condition
for noninteraction by Buchholz \cite{Buchholz}. We prove this
based on an observation that the maps which give asymptotic fields
are conditional expectations between appropriate
algebras. Another consequence of this observation is that the maps which give
in- and out-asymptotic fields
are the conditional expectations onto the chiral components in M\"obius covariant theory,
hence they coincide.

This paper is organized as follows. In Section \ref{preliminaries} we
recall the basic notions of Poincar\'e covariant nets with various higher symmetries
in two-dimensional spacetime and the scattering theory of massless ``waves'' studied in \cite{Buchholz}.
In Section \ref{noninteraction} we demonstrate that these waves have
always trivial scattering matrix in M\"obius covariant nets.
This proof is based on the representation theory of M\"obius group, and one derives
M\"obius symmetry from Bisognano-Wichmann property, dilation covariance and asymptotic completeness
with respect to waves.
In Section \ref{subspace} the chiral components
are defined following \cite{Rehren}.
They turn out to generate all the waves from
the vacuum. In Section \ref{characterization}, under Bisognano-Wichmann property
and Haag duality, we show that asymptotic fields are given by conditional expectations and
that a Poincar\'e covariant net is asymptotically
complete (with respect to waves) and noninteracting if and only if it is isomorphic to a chiral
M\"obius net. In Section \ref{af-in-mob} we show that in- and out-asymptotic
fields coincide in M\"obius covariant nets.
In Section \ref{concluding} we discuss open problems and perspectives.
In \ref{expectation} we collect fundamental facts about conditional expectations
and in \ref{chiral-components} remarks about various definitions of
chiral component are given.

\section{Preliminaries}\label{preliminaries}
\subsection{Conformal nets}\label{conformalnets}

In algebraic QFT, we consider nets of observables.
Let us briefly recall the definitions.
The two-dimensional Minkowski space $\RR^2$ is represented as a product of
two lightlines $\RR^2 = \L_+\times\L_-$, where $\L_\pm := \{(a_0,a_1)\in \RR^2: a_0\pm a_1=0\}$
are the positive and the negative lightlines.
The fundamental group of spacetime symmetry is the (proper orthochronous) Poincar\'e group
$\poincare$, which is generated by translations and Lorentz boosts.

Let $\regions$ be the family of open bounded regions in $\RR^2$.
A {\bf (local) Poincar\'e covariant net} $\A$ assigns to $O \in \regions$
a von Neumann algebra $\A(O)$ on a common separable Hilbert space $\H$
satisfying the following conditions:
\begin{enumerate}
\item[(1)] {\bf Isotony.} If $O_1 \subset O_2$, then $\A(O_1) \subset \A(O_2)$.\label{isotony}
\item[(2)] {\bf Locality.} If $O_1$ and $O_2$ are spacelike separated, then $[\A(O_1),\A(O_2)] = 0$.
\item[(3)] {\bf Additivity.} If $O = \bigcup_i O_i$, then $\A(O) = \bigvee_i \A(O_i)$.
\item[(4)] {\bf Poincar\'e covariance.} There exists a strongly continuous unitary
representation $U$ of the Poincar\'e group $\poincare$
such that
\begin{equation*}
U(g)\A(O)U(g)^* = \A(gO), \mbox{ for } g \in \poincare.
\end{equation*}
\item[(5)]{\bf Positivity of energy.} The joint spectrum of the translation subgroup
in $\poincare$ in the representation $U$ is contained
in the forward lightcone $V_+ := \{(p_0,p_1)\in \RR^2: p_0+p_1\ge 0, p_0-p_1\ge 0\}$.
\item[(6)] {\bf Existence of the vacuum.} There is a unique (up to a phase) unit vector $\Omega$ in
$\H$ which is invariant under the action of $U$,
and cyclic for $\bigvee_{O \in \regions} \A(O)$.
\end{enumerate}

 From these assumptions, the following properties automatically follow \cite{Baumgaertel}.
\begin{itemize}
\item[(7)] {\bf Reeh-Schlieder property.} The vector $\Omega$ is cyclic and separating for each $\A(O)$.
\item[(8)] {\bf Irreducibility.} The von Neumann algebra $\bigvee_{O \in \regions} \A(O)$ is equal to $B(\H)$.
\end{itemize}

We identify the circle $S^1$ as the one-point compactification of the real line $\RR$
by the Cayley transform:
\[
t = i\frac{z-1}{z+1} \Longleftrightarrow z = -\frac{t-i}{t+i}, \phantom{...}t \in \RR,
\phantom{..}z \in S^1 \subset \CC.
\]
The M\"obius group $\psl2r$ acts on $\RR \cup \{\infty\}$ by the linear fractional
transformations, hence it acts on $\RR$ locally (see \cite{BGL} for local actions).
Then the group $\2dmob$ acts
locally on $\RR^2$, where $\overline{\psl2r}$ is the universal covering group
of $\psl2r$. Note that the group $\2dmob$ contains
translations, Lorentz boosts and dilations, so in particular it includes the Poincar\'e
group $\poincare$.  We refer to \cite{KL03} for details.

Let $\A$ be a Poincar\'e covariant net. 
If the representation $U$ of $\poincare$ (associated to the net $\A$) extends to $\2dmob$
such that for any open region $O$ there is a small neighborhood $\U$ of the unit element
in $\2dmob$ such that $gO \subset \RR^2$ and it holds that
\begin{equation*}
U(g)\A(O)U(g)^* = \A(gO), \mbox{ for } g \in \U,
\end{equation*}
then we say that $\A$ is a {\bf M\"obius covariant net}.

If the net $\A$ is M\"obius covariant,
then it extends to a net on the Einstein cylinder $\E := \RR\times S^1$ \cite{KL03}.
On $\E$ one can define a natural causal structure which extends the one on $\RR^2$
(see \cite{LM}). We take a coordinate system on $\E$ used in \cite{Rehren}:
Let $\RR\times\RR$ be the universal cover of $S^1\times S^1$. The cylinder $\E$
is obtained from $\RR\times\RR$ by identifying points $(a,b)$ and $(a+2\pi,b-2\pi) \in \RR\times\RR$.
Any double cone of the form $(a,a+2\pi)\times(b,b+2\pi) \subset \RR\times \RR$
represents a copy of the Minkowski space. The causal complement of a double cone
$(a,c)\times(b,d)$, where $0<c-a<2\pi,0<d-b<2\pi$, is $(c,a+2\pi)\times(d-2\pi,b)$
or equivalently $(c-2\pi,a)\times(d,b+2\pi)$. If $O$ is a double cone, we denote
the causal complement by $O'$.
For an interval $I = (a,b)$, we denote by $I^+$ the interval $(b,a+2\pi) \subset\RR$
and by $I^-$ the interval $(b-2\pi,a)\subset \RR$.

Furthermore, it is well-known that, from M\"obius covariance, the following properties automatically follow
(see \cite{BGL}):
\begin{enumerate}
\item[(9M)] {\bf Haag duality in $\E$.} For a double cone $O$ in $\E$ it holds that $\A(O)' = \A(O')$,
where $O'$ is defined in $\E$ as above.
\item[(10M)] {\bf Bisognano-Wichmann property in $\E$} For a double cone $O$ in $\E$, the modular automorphism
group $\Delta_O^{it}$ of $\A(O)$ with respect to the vacuum state $\omega := \<\Omega, \cdot\Omega\>$
equals to $U(\Lambda^O_t)$ where $\Lambda^O_t$ is a one-parameter group in $\2dmob$
which preserves $O$ (see \cite{BGL} for concrete expressions).
\end{enumerate}

We denote by $\diffr$ the group of diffeomorphisms of $S^1$ which preserve the point $-1$.
If we identify $S^1 \setminus \{-1\}$ with $\RR$, this can be considered as a group of
diffeomorphisms of $\RR$
\footnote{Note that not all diffeomorphisms of $\RR$ extend to diffeomorphisms of $S^1$,
hence the group $\diffr$ is {\em not} the group of all the diffeomorphisms of $\RR$.
However, this notation is common in the context of conformal field theory.}.
The Minkowski space $\RR^2$ can be identified with a double cone in $\E$.
The group $\diffr\times\diffr$ acts on $\RR^2$ and this action extends to $\E$ by
periodicity. The group generated by this action of $\diffr\times\diffr$ and
the action of $\2dmob$ (which acts on $\E$
through quotient by the relation $(r_{2\pi},r_{-2\pi}) = (\id,\id)$ \cite{KL03}) is denoted by $\conf$.
Explicitly, $\conf$ is isomorphic to the quotient group of $\overline{\diffs1}\times\overline{\diffs1}$
by the normal subgroup generated by $(r_{2\pi},r_{-2\pi})$, where $\overline{\diffs1}$ is the
universal covering group of $\diffs1$ (note that $r_{2\pi}$ is an element
in the center of $\overline{\diffs1}$).

A M\"obius covariant net is said to be {\bf conformal} if the representation $U$ further extends to
a projective representation of $\conf$ such that
\[
U(g)\A(O)U(g)^* = \A(gO),\mbox{ for } g \in \diffr\times\diffr,
\]
and if it holds that
\[
U(g)xU(g)^* = x
\]
for $x \in \A(O)$, where $O$ is a double cone and $g \in \diffr\times\diffr$ has
a support disjoint from $O \subset \RR^2$.

\begin{proposition}\label{chiral-net-contains-virasoro}
If the net $\A$ is conformal, the intersection $\bigcap_{J} \A(I\times J)$ contains representatives of
diffeomorphisms of the form $g_+\times \id$ where $\supp(g_+) \subset I$,
\end{proposition}
\begin{proof}
If $g$ is a diffeomorphism of the form $g_+\times \id$ and $\supp(g_+) \subset I$,
then $U(g)$ commutes with $\A(I^+\times J)$ for arbitrary $J$, thus
Proposition follows by the Haag duality in $\E$.
\end{proof}

In the rest of the present paper, conformal covariance will not be
assumed except \ref{chiral-components}, although a major part of examples of
M\"obius covariant nets is in fact conformal.

If it holds that $\A(O_1)\vee\A(O_2) = \A(O)$ where $O_1$ and $O_2$ are
the two components of the causal complement (in $O$)
of an interior point of a double cone $O$, we say that $\A$ is
{\bf strongly additive}. This implies the {\bf chiral additivity} \cite{Rehren},
namely that $\A(I\times J_1)\vee\A(I\times J_2) = \A(I\times J)$ if
$J_1$ and $J_2$ are obtained from $J$ by removing an interior point.

We recall that $\A(O)$ is interpreted as the algebra of observables measured
in a spacetime region $O$. A typical example of a conformal net is constructed
in the following way: If we have a local conformal field $\Psi$, namely
an operator-valued distribution, then we define $\A(O)$ as the von Neumann
algebra generated by $e^{i\Psi(f)}$ where the support of $f$ is included in $O$.
But our framework does not assume the existence of any field.
Indeed, there are examples of nets for which no local field description is at hand \cite{Lechner}. 
Thus the algebraic approach is
more general than the conventional one. It also
provides a natural scattering theory, as we recall in the next section.

\subsection{Scattering theory of waves}\label{scatteringtheory}
Here we summarize the scattering theory of massless two-dimensional models
established in \cite{Buchholz}. This theory is stated
in terms of Poincar\'e covariant nets of observables.

Let us denote by $T(a) := U(\tau(a))$ the representative of spacetime translation $\tau(a)$
by $a \in \RR^2$.
Furthermore, we denote the lightlike translations by $T_\pm(t) := T(t,\pm t)$.
Let $\axb$ denote the subgroup of $\psl2r$ generated by (one-dimensional)
translations and dilations.
Note that $\axb$ is simply connected, hence it can be considered as a subgroup
of $\overline{\psl2r}$.
As will be seen in the following, the representation $U$ of $\2dmob$
restricted to $\axb\times\axb$
has typically a big multiplicity in M\"obius covariant theories.
The subspaces $\H_+ = \{\xi \in \H: T_+(t)\xi = \xi \mbox{ for all } t \}$ and
$\H_- = \{\xi \in \H: T_-(t)\xi = \xi \mbox{ for all  } t \}$
are referred to as the spaces of waves with positive and negative momentum, respectively.
Let $P_\pm$ be the orthogonal projections onto $\H_\pm$, respectively.

Let $x$ be a local operator, i.e., an element of $\A(O)$ for some $O$. 
We set $x(a) := T(a)xT(a)^*$ for $a\in\RR^2$ and
consider a family of operators parametrized by $\T$:
\[
x_\pm(h_\T) := \int dt\, h_\T(t) x(t,\pm t),
\]
where $h_\T(t) = |\T|^{-\e}h(|\T|^{-\e}(t-\T))$, $0<\e<1$ is a constant,
$\T \in \RR$ and $h$ is a nonnegative symmetric smooth function on $\RR$ such that
$\int dt\, h(t) = 1$.

\begin{lemma}[\cite{Buchholz} Lemma 1,2,3]\label{asymptotic-field}
Let $x$ be a local operator.
Then the limit $\Phi^\tin_\pm(x) := \underset{\T\to-\infty}\slim \, x_\pm(h_\T)$ exists
and it holds that
\begin{itemize}
\item $\Phi^\tin_\pm(x)\Omega = P_\pm x\Omega$.
\item $\Phi^\tin_\pm(x)\H_\pm \subset \H_\pm$.
\item $\Ad U(g) \Phi^\tin_\pm(x) = \Phi^\tin_\pm(\Ad U(g)(x))$, where $g \in \poincare$.
\end{itemize}
Furthermore, the limit $\Phi^\tin_\pm(x)$ depends only on $P_\pm x\Omega$, respectively.
We call these limit operators the ``incoming asymptotic fields''.
It holds that $[\Phi^\tin_+(x),\Phi^\tin_-(y)]=0$ for arbitrary local $x$ and $y$.

Similarly one defines the ``outgoing asymptotic fields'' by
$\Phi^\tout_\pm(x) := \underset{\T\to\infty}\slim \,x_\pm(h_\T)$
\end{lemma}
\begin{remark}\label{commutation-asymptotic-field}
As the asymptotic field is defined as the limit of local operators, it still has
certain local properties. For example, let $O_+$ and $O_0$ be two regions such that
$O_+$ stays in the future of $O_0$ and $x \in O_+, y \in O_0$. Then it holds
that $[\Phi^\tin_\pm(x),y] = 0$, since for a negative $\T$ with sufficiently large absolute value,
$x_\pm(h_\T)$ lies in the spacelike complement of $y$.
Similar observations apply also to $\Phi^\tout_\pm$.
\end{remark}

Lemma \ref{asymptotic-field} 
captures the dispersionless kinematics of elementary excitations
in two-dimensional massless theories: since $\Phi^\tin_\pm(x)\H_\pm \subset \H_\pm$, by composing
two waves travelling to the right we obtain again a wave travelling to the right. Thus waves are,
in general, composite objects, associated with reducible representations
of the Poincar\'e group.
Moreover, it follows that collision states of waves may contain at most two excitations:
One wave with positive momentum and the other with negative momentum.

Let us now construct these collision states:
For $\xi_\pm \in \H_\pm$, there are sequences of local operators $\{x_{\pm,n}\}$ such that 
$\underset{n\to\infty}\slim \,P_\pm x_{\pm,n}\Omega = \xi_\pm$ and
Using these sequences let us define collision states following \cite{Buchholz} (see also \cite{DT10-1}):
\begin{eqnarray*}
\xi_+\timesi\xi_- &=& \underset{n\to\infty}\slim \,\Phi^\tin_+(x_{+,n})\Phi^\tin_-(x_{-,n})\Omega\\
\xi_+\timeso\xi_- &=& \underset{n\to\infty}\slim \,\Phi^\tout_+(x_{+,n})\Phi^\tout_-(x_{-,n})\Omega\\
\end{eqnarray*}
We interpret $\xi_+\timesi\xi_-$ (respectively $\xi_+\timeso\xi_-$) as the incoming
(respectively outgoing) state which describes two non-interacting waves
$\xi_+$ and $\xi_-$.
These asymptotic states have the following natural properties.
\begin{lemma}[\cite{Buchholz} Lemma 4]\label{collision-states}
For the collision states $\xi_+\timesi\xi_-$ and $\eta_+\timesi\eta_-$ it holds that
\begin{enumerate}
\item $\<\xi_+\timesi\xi_-, \eta_+\timesi\eta_-\> = \<\xi_+,\eta_+\>\cdot\<\xi_-,\eta_-\>$.
\item $U(g) (\xi_+\timesi\xi_-) = (U(g)\xi_+)\timesi(U(g)\xi_-)$ for all $g\in\poincare$.
\end{enumerate}
And analogous formulae hold for outgoing collision states.
\end{lemma}

Furthermore, we define the spaces of collision states: Namely,
we let $\H^\tin$ (respectively $\H^\tout$) be the subspace generated by $\xi_+\timesi\xi_-$
(respectively $\xi_+\timeso\xi_-$).
 From the Lemma above, we see that the following map
\[
S: \xi_+\timeso\xi_- \longmapsto \xi_+\timesi\xi_-
\]
is an isometry. The operator $S: \H^\tout \to \H^\tin$ is called
the {\bf scattering operator} or the {\bf S-matrix}.
We say the waves in $\A$ are {\bf interacting} if $S$ is not a constant
multiple of the identity operator on $\H^\tout$.
The purpose of this paper is to show that 
$S = \1$ on $\H^\tout$ for M\"obius covariant nets and to determine $\H^\tout = \H^\tin$
in terms of chiral observables (see Section \ref{subspace}).
As a corollary one observes that a M\"obius covariant net is chiral if and only
if it is {\bf asymptotically complete (with respect to waves)}, i.e. $\H^\tout=\H^\tin=\H$.
We remark that this notion of asymptotic completeness refers only to
massless excitations. If one considers the massive free field, all the asymptotic fields
considered here reduce to multiples of the identity.
Throughout this article, we are concerned only with waves.

Moreover, we show in Section \ref{characterization}
that if a net is Poincar\'e covariant and asymptotically complete,
then it is noninteracting if and only if it is a chiral M\"obius covariant net
(see Section \ref{chiralnets}) and in Section \ref{af-in-mob} that in- and out-
asymptotic fields coincide in a (possibly non-chiral) M\"obius covariant net.

\section{Noninteraction of waves}\label{noninteraction}
\subsection{Representations of the spacetime symmetry group}\label{representation-theory}
As a preliminary for the proof of the main result, we need to examine the
structure of representations of the group generated by translations and dilations.

Recall that we denote by $\axb$ the subgroup of $\psl2r$ generated by (one-dimensional)
translations and dilations. 
The group $\axb$ is simply connected, hence it can be considered as a subgroup
of $\overline{\psl2r}$.
The direct product $\axb\times\axb\subset \psl2r\times\psl2r$
is the group of (two-dimensional) translations, Lorentz boosts and dilations.
For the later use, we only have to consider representations of $\axb\times\axb$
which extend to  positive-energy representations of $\2dmob$.

Recall further that irreducible positive-energy representations of $\overline{\psl2r}$ are classified by
a nonnegative number $l$, which is the lowest eigenvalue of the generator of 
(the universal covering of) the group of rotations (see \cite{Longo08}).
We claim that irreducible representations of
$\2dmob$ are classified by pairs of nonnegative numbers $l_\L,l_\R$.
Indeed, we can take the G\r{a}rding domain $\D$ since $\2dmob$
is a finite dimensional Lie group.
Furthermore, if a representation is irreducible, then the center of the group
must act as scalars. From this it follows that the joint spectrum of generators of
left and right rotations is discrete and each point must have positive components by
the assumed positivity of energy.
The same argument as in \cite{Longo08} shows that an eigenvector with minimal eigenvalues
of rotations generates an irreducible representation,
hence irreducible representations are classified by this pair of minimal eigenvalues.
Conversely, all of these representations are realized by product representations.
Let us sum up these observations:
\begin{proposition}\label{classification-psl2r}
All the irreducible representations of $\2dmob$ are
completely classified by pairs of nonnegative numbers $(l_\L,l_\R)$.
A representation with a given $(l_\L,l_\R)$ is unitarily equivalent to the product of
representations of $\overline{\psl2r}$ with lowest weights $l_\L,l_\R$
($l_\L = 0$ or $l_\R = 0$ correspond to the trivial representation).
A vector in any of these irreducible representations is invariant
under the subgroup $\overline{\psl2r}\times \id$ if and only if
it is invariant under $\t_0 \times \id$, where $\t_0$ is
the translation subgroup of $\overline{\psl2r}$ (and the same holds for
the right component).
\end{proposition}

We know that if $l \neq 0$ then the restriction of the representation to
$\axb$ is the unique strictly positive-energy representation \cite{Longo08}
(here ``positive-energy'' means that the generator of translations is positive).
As a consequence of Proposition \ref{classification-psl2r}, we can classify positive-energy irreducible
representations of $\axb\times\axb$ which appear in M\"obius covariant nets.
\begin{corollary}\label{classification-axb}
Let $\i$ and $\rho$ be the trivial and the unique strictly positive-energy representation
of $\axb$ respectively.
Any irreducible positive-energy representation of $\axb\times\axb$ which extends to
$\2dmob$ is one of the following four representations.
\begin{itemize}
\item $\i\otimes\i$,
\item $\rho\otimes\i$,
\item $\i\otimes\rho$,
\item $\rho\otimes\rho$.
\end{itemize}
Any (possibly reducible) representation of $\axb\times\axb$ extending
to $\2dmob$ is a direct sum
of copies of the above four representations.
\end{corollary}
\begin{proof}
The first part of the statement follows directly from Proposition \ref{classification-psl2r}.
The second part is a consequence of the general result
(for example, see \cite[Sections 8.5 and 18.7]{Dixmier})
that any continuous unitary representation (on a separable Hilbert space)
of a (separable) locally compact group is unitarily equivalent
to a direct integral of irreducible representations. Since by assumption the
given representation extends to $\2dmob$,
it decomposes into a direct integral, and the components have positive-energy
almost everywhere. Hence they are classified by $(l_\L,l_\R)$ and when restricted to $\axb\times\axb$
they fall into irreducible representations listed above.
Since the integrand takes only four different values
(up to unitary equivalence), the direct integral reduces to a direct sum.
\end{proof}

\subsection{Proof of noninteraction}\label{proof-noninteraction}
As waves are defined in terms of representations of translations,
we need to analyse the representation $U$. We continue to use notations
 from the previous section. A net $\A$ in this section is always assumed to
be M\"obius covariant.

The representation $\rho$ of $\axb$ does not admit any nontrivial
invariant vector with respect to (one-dimensional) translations.
The subgroup of dilations is noncompact (isomorphic to $\RR$) and for any vector $\xi$ 
in the representation space of $\rho$
it holds that $\rho(\d_s)\xi$ tends weakly $0$ as $s \to \pm \infty$,
where $\d_s$ represents the group element of dilation by $e^s$.

\begin{remark}\label{extension-to-mob}
At this point we  use the assumed covariance
under the action of the two dimensional M\"obius group
$\2dmob$.
If we assume only the dilation covariance (as in \cite{BF}),
the present author is not able to exclude the possibility
of occurrence of a representation of $\axb$ which is trivial only on translations in general.
As we will see, the absence of such representations is essential to identify all the waves in the
relevant representation space.

But if one assumes Bisognano-Wichmann property and asymptotic completeness in addition,
it is possible to show that the representation of the spacetime symmetry extends to the
M\"obius group: As observed in \cite{Tanimoto11-3}, for an asymptotically complete
Poincar\'e covariant net with
Bisognano-Wichmann property, one can define the asymptotic net which is chiral M\"obius
covariant. The representation of the M\"obius group is a natural extension of the given
representation of the Poincar\'e group given through the Bisognano-Wichmann property.
Their actions on asymptotic fields are determined
by the boosts, hence the representation extends also the given representation of dilation.
Summing up, under Bisognano-Wichmann property and asymptotic completeness, the representation
of the Poincar\'e group and dilation extends to the M\"obius group.
\end{remark}

Among the four irreducible positive-energy representations of $\axb\times\axb$
(see Corollary~\ref{classification-axb}),
only $\i\otimes\i$ contains a nonzero invariant vector
with respect to two-dimensional translations. The representation space of $\i\otimes\rho$
consists of invariant vectors with respect to positive-lightlike translations
but contains no nonzero invariant vectors with respect to negative-lightlike
translations. An analogous statement holds for $\rho\otimes\i$.
The representation $\rho\otimes\rho$ contains no nonzero invariant vectors,
neither with respect to negative- nor positive-lightlike translations.

Let us consider the representation $U$ of $\2dmob$
associated with a M\"obius covariant net $\A$.
The restriction of $U$ to $\axb\times\axb$ is a direct sum
of copies of representations which appeared in Corollary \ref{classification-axb}.
By the uniqueness of the vacuum, the representation $\i\otimes\i$ appears only
once. Waves of positive (respectively negative) momentum correspond precisely
to $\rho\otimes\i$ (respectively $\i\otimes\rho$).
 From these observations, it is straightforward to see the following.
\begin{lemma}\label{dilate}
Let us denote by $P$ the spectral measure of the
representation $T = U|_{\RR^2}$ of translations.
Each of the following spectral subspaces of $T$ carries the multiple of one of the
irreducible representations in Corollary \ref{classification-axb} (the correspondence
is the order of appearance)
\begin{itemize} 
\item $Q_0 := P(\{(0,0)\})$,
\item $Q_\L := P(\{(a_0,a_1): a_0 = a_1, a_0> 0\})$,
\item $Q_\R := P(\{(a_0,a_1): a_0 = - a_1, a_0 > 0\})$,
\item $Q_{\L,\R} := P(\{(a_0,a_1): a_0 > a_1, a_0 > - a_1\})$.
\end{itemize}
Let $\d^\L$ be the dilation in the left-component of $\2dmob$. Then
for any vector $\xi \in \H$, $\underset{s\to 0}\wlim \,U(\d^\L_s)\xi = (Q_\R+Q_0)\xi$.
Similarly for the dilation in the right component $\d^\R$ we have
$\underset{s\to 0}\wlim\, U(\d^\R_s)\xi = (Q_\L+Q_0)\xi$.
Furthermore, it holds that $Q_\L+Q_0 = P_+, Q_\R+Q_0 = P_-$ (see Section \ref{scatteringtheory}
for definitions)
\end{lemma}

After this preparation we proceed to our main result:
\begin{theorem}\label{noninteraction-bf}
Let $\A$ be a M\"obius covariant net.
We have the equality $\xi_+\timesi\xi_- = \xi_+\timeso\xi_-$ for any
pair $\xi_+ \in \H_+$ and $\xi_- \in \H_-$. In particular,
such waves do not interact and we have $\H^\tout = \H^\tin$.
\end{theorem}
\begin{proof}
We show the equality $\<\xi_+\timesi\xi_-, \eta_+\timeso\eta_-\>
= \<\xi_+\timesi\xi_-, \eta_+\timesi\eta_-\>$ for any $\xi_+,\eta_+ \in \H_-$
and $\xi_-,\eta_- \in \H_-$. This is in fact enough for the first statement,
since we know that $\|\eta_+\timeso\eta_-\| = \|\eta_+\timesi\eta_-\|$.
As a particular case we have $\<\eta_+\timesi\eta_-, \eta_+\timeso\eta_-\>
= \<\eta_+\timesi\eta_-, \eta_+\timesi\eta_-\>$, which is possible only if
$\eta_+\timeso\eta_- = \eta_+\timesi\eta_-$.

Obviously it suffices to show the equality for a dense set of vectors
in $\H_+$ and $\H_-$. Let us take three double cones $O_+, O_0, O_-$ which
are timelike separated in this order, more precisely $O_0$ stays in the future of $O_-$
and in the past of $O_+$, and assume that $O_0$ is a neighborhood
of the origin. We choose elements $x_+ \in \A(O_+)$ and $y_+,y_- \in \A(O_-)$.
We take a self-adjoint element $b \in \A(O_0)$ and set $b_s := \Ad(U(\d^\L_s))(b)$
for $s < 0$.
Then $\{b_s\}$ are still contained in $\A(O_0)$. 
We set:
\begin{align*}
\xi_+&:=\Phi^\tin_+(x_+)\Omega,
&\xi_- &:= \underset{s\to 0}\wlim\, b_s\Omega = \underset{s\to 0}\wlim\, U(\d^\L_s)b\Omega,\\
\eta_+&:=\Phi^\tout_+(y_+)\Omega,
&\eta_-&:=\Phi^\tout_-(y_-)\Omega,\\ 
\zeta_-&:=\Phi^\tout_-(y_-^*)\Omega = \Phi^\tout_-(y_-)^*\Omega.
\end{align*}
Note that $b_s$ commutes with $\Phi^\tin_+(x_+)$, $\Phi^\tout_+(y_+)$
and $\Phi^\tout_-(y_-)$ since $\Phi^\tin$ and $\Phi^\tout$ are
defined as strong limits of local operators and from some point
they are spacelike separated (see Remark \ref{commutation-asymptotic-field}).
Note also that
$\Phi^\tin_+(x_+)\Omega = P_+ x_+\Omega, \Phi^\tout_+(y_+)\Omega = P_+ y_+\Omega,
\Phi^\tout_-(y_-)\Omega = P_- y_-\Omega$ and we have $\lim_s b_s\Omega = P_- b\Omega$
by Lemma \ref{dilate}.

We see that
\begin{eqnarray*}
\<\xi_+\timesi\xi_-, \eta_+\timeso\eta_-\>
&=& \<\Phi^\tin_+(x_+)(\underset{s\to 0}\wlim\,b_s\Omega), \Phi^\tout_+(y_+)\Phi^\tout_-(y_-)\Omega\> \\
&=& \lim_s \<\Phi^\tin_+(x_+)b_s\Omega, \Phi^\tout_+(y_+)\Phi^\tout_-(y_-)\Omega\> \\
&=& \lim_s \<\Phi^\tout_-(y_-^*)\Phi^\tin_+(x_+)\Omega, \Phi^\tout_+(y_+)b_s\Omega\>, 
\end{eqnarray*}
where we used Remark \ref{commutation-asymptotic-field} in the 3rd line.
Continuing the calculation, with the help of the definition of asymptotic fields,
this can be transformed as
\begin{eqnarray*}
\<\xi_+\timesi\xi_-, \eta_+\timeso\eta_-\>
&=& \<\Phi^\tout_-(y_-^*)\Phi^\tin_+(x_+)\Omega, \Phi^\tout_+(y_+)(\underset{s\to 0}\wlim\,b_s\Omega)\> \\
&=& \<\Phi^\tout_-(y_-^*)\xi_+, \Phi^\tout_+(y_+)\xi_-\> \\
&=& \<\xi_+\timeso\zeta_-, \eta_+\timeso\xi_-\> \\
&=& \<\xi_+,\eta_+\>\cdot\<\zeta_-,\xi_-\> \\
&=& \<\xi_+,\eta_+\>\cdot \<\Phi^\tout_-(y_-^*)\Omega,(\underset{s\to 0}\wlim\,b_s\Omega)\> \\
&=& \<\xi_+,\eta_+\>\cdot \<(\underset{s\to 0}\wlim\,b_s\Omega),\Phi^\tout_-(y_-)\Omega\> \\
&=& \<\xi_+,\eta_+\>\cdot \<\xi_-,\eta_-\> \\
&=& \<\xi_+\timesi\xi_-, \eta_+\timesi\eta_-\>,
\end{eqnarray*}
where the 6th equality follows from Remark \ref{commutation-asymptotic-field}
and the self-adjointness
of $b$, the 4th and 8th equalities follow from Lemma \ref{collision-states}.
This equation is  linear with respect to $b$ (which is implicitly contained
in $\xi_-$), hence it holds for any $b \in \A(O_0)$.

By the Reeh-Schlieder property, each set of vectors of the forms above is dense
in $\H_+$ and $\H_-$, respectively. Thus the required equality is obtained
for dense subspaces and this concludes the proof.
\end{proof}

The proof of this Theorem uses only the fact that $\A$ is Poincar\'e-dilation covariant
and that the representation of the Poincar\'e-dilation group extends to
the M\"obius group. Putting together with Remark \ref{extension-to-mob}, we obtain
\begin{corollary}\label{ac-dilation}
If a dilation-covariant net $\A$ satisfies Bisognano-Wichmann property and
asymptotic completeness, then the waves in $\A$ are not interacting.
\end{corollary}

\section{Subspace of collision states of waves}\label{subspace}
It has been shown by Rehren that
any M\"obius covariant net contains the maximal chiral subnet, consisting of observables
localized on the lightrays \cite{Rehren}. Here we show that the
vectors generated by such observables from the vacuum exhaust the subspace of
collision states.  With this information at hand, we provide an alternative proof
of noninteraction of waves and show that a M\"obius covariant field theory is asymptotically
complete if and only if it is chiral.

\subsection{Preliminaries on chiral nets}\label{chiralnets}
In this section we discuss a fundamental class of examples
of two-dimensional M\"obius covariant nets, namely chiral theories.
A chiral theory is obtained by a tensor product construction
from two nets of von Neumann algebras on a circle $S^1$, defined below,
and each of these nets is referred to as {\bf the chiral component} of the theory.

An open nonempty connected nondense subset $I$ of the circle $S^1$ is called
an interval.
A {\bf (local) M\"obius covariant net $\A_0$ on $S^1$} assigns to each interval
a von Neumann algebra $\A_0(I)$ on a fixed separable Hilbert space $\H_0$
satisfying the following conditions:
\begin{enumerate}
\item[(1)] {\bf Isotony.} If $I_1 \subset I_2$, then $\A_0(I_1) \subset \A_0(I_2)$.
\item[(2)] {\bf Locality.} If $I_1 \cap I_2 = \emptyset$, then $[\A_0(I_1),\A_0(I_2)] = 0$.
\item[(3)] {\bf M\"obius covariance.} There exists a strongly continuous unitary
representation $U_0$ of the M\"obius group $\psl2r$ such that
for any interval $I$ it holds that
\begin{equation*}
U_0(g)\A_0(I)U_0(g)^* = \A_0(gI), \mbox{ for } g \in \psl2r.
\end{equation*}
\item[(4)]{\bf Positivity of energy.} The generator of the one-parameter subgroup of
rotations in the representation $U_0$ is positive.
\item[(5)] {\bf Existence of the vacuum.} There is a unique (up to a phase) unit vector $\Omega_0$ in
$\H_0$ which is invariant under the action of $U_0$,
and cyclic for $\bigvee_{I \Subset S^1} \A_0(I)$.
\end{enumerate}
Among consequences of these axioms are (see \cite{FG})
\begin{enumerate}
\item[(6)] {\bf Reeh-Schlieder property.} The vector $\Omega_0$ is cyclic and separating for each $\A_0(I)$.
\item[(7)] {\bf Additivity.} If $I = \bigcup_i I_i$, then $\A_0(I) = \bigvee_i \A_0(I_i)$.
\item[(8)] {\bf Haag duality in $S^1$.} For an interval $I$ it holds that $\A_0(I)' = \A_0(I')$,
where $I'$ is the interior of the complement of $I$ in $S^1$.
\item[(9)] {\bf Bisognano-Wichmann property.} The modular group $\Delta_0^{it}$ of $\A_0(\RR_+)$
with respect to $\Omega_0$ is equal to $U_0(\delta(-2\pi t))$, where
$\delta$ is the one-parameter group of dilations.
\end{enumerate}
It is known that the positivity of energy is equivalent to the positivity of
the generator of translations \cite{Longo08}.

We say that $\A_0$ is {\bf strongly additive} if it holds that
$\A_0(I) = \A_0(I_1)\vee\A_0(I_2)$, where $I_1$ and $I_2$ are intervals obtained
by removing an interior point of $I$.

Let $\diffs1$ be the group of orientation-preserving
diffeomorphisms of the circle $S^1$. This group naturally includes $\psl2r$.
If the representation $U_0$ associated to a M\"obius covariant net $\A_0$
extends to a projective unitary representation of $\diffs1$ such that
for any interval $I$ and $x \in \A_0(I)$ it holds that
\begin{gather*}
U_0(g)\A_0(I)U_0(g)^* = \A_0(gI), \mbox{ for } g \in \diffs1,\\
U_0(g)xU_0(g)^* = x, \mbox{ if } \supp(g) \subset I^\prime,
\end{gather*}
then $\A_0$ is said to be a {\bf conformal net on $S^1$} or to be
{\bf diffeomorphism covariant} ($\supp(g) \subset I^\prime$
means that $g$ acts identically on $I$).

Let $\A_0$ be a M\"obius covariant net on $S^1$. As in Section \ref{conformalnets},
we identify $S^1$ and $\RR\cup\{\infty\}$ by the Cayley transform.
Under this identification, for an interval $I \Subset \RR$ we write
$\A_0(I)$.

Let $\A_\pm$ be two M\"obius covariant nets on $S^1$ defined on the Hilbert spaces $\H_\pm$ 
with the vacuum vectors $\Omega_\pm$ and the representations of $U_\pm$.
We define a two-dimensional net $\A$ as follows:
Let $\L_\pm := \{(t_0,t_1) \in \RR^2: t_0\pm t_1 = 0\}$ be two lightrays.
For a double cone $O$ of the form $I\times J$ where
$I\subset \L_+$ and $J\subset \L_-$, we set $\A(O) = \A_+(I)\otimes\A_-(J)$.
For a general open region $O \subset \RR^2$, we set $\A(O) := \bigvee_{I\times J} \A(I\times J)$
where the union is taken among intervals such that $I\times J \subset O$.
If we take the vacuum vector as $\Omega := \Omega_+\otimes\Omega_-$ and
define the representation $U$ of $\psl2r\times\psl2r$
by $U(g_+\times g_-) := U_+(g_1)\times U_-(g_2)$, it is easy to see that
all the conditions for M\"obius covariant net follow from
the corresponding properties of nets on $S^1$.
We say that such $\A$ is {\bf chiral}.
If $\A_\pm$ are conformal, then the representation $U$ naturally extends to a projective representation
of $\diffs1\times\diffs1$. Hence $\A$ is a two-dimensional conformal net.

\subsection{The maximal chiral subnet and collision states}\label{maximal}
As we have seen in Section \ref{chiralnets}, from a pair of M\"obius covariant nets
on $S^1$ we can construct a two-dimensional M\"obius covariant net.
In this section we explain a converse procedure: Namely, starting with
a two-dimensional M\"obius covariant net $\A$, we find a pair of M\"obius covariant nets $\A_\pm$
on $S^1$ which are maximally contained in $\A$. In general, such a chiral
part is just a subnet of the original net.
Moreover, we show that the subspace generated
by this subnet from the vacuum coincides with the subspace of collision states of waves.
It follows that a M\"obius covariant net is asymptotically complete if and only if it is chiral.

It is possible to define  chiral components in several ways.
We follow the definition by Rehren \cite{Rehren}.
Recall that the two-dimensional M\"obius group $\2dmob$
is a direct product of two copies of the universal covering group of $\psl2r$.
We write this as $\widetilde{G}_\L\times\widetilde{G}_\R$, where
$\widetilde{G}_\L$ and $\widetilde{G}_\R$ are copies of $\overline{\psl2r}$
\footnote{Generally, the symbol $\widetilde{G}$ is used to indicate the universal covering
group for a group $G$, but for $\psl2r$ it is customary to use
the notation $\overline{\psl2r}$ for its universal cover.}.
\begin{definition}\label{chiral-component}
For a two-dimensional M\"obius net $\A$ we define
nets of von Neumann algebras on $\RR$ by the following:
For an interval $I \subset \RR$ we set the von Neumann algebras
\begin{gather*}
\A_\L(I) := \A(I\times J)\cap U(\widetilde{G}_\R)',\\
\A_\R(J) := \A(I\times J)\cap U(\widetilde{G}_\L)'.
\end{gather*}
The definition of $\A_\L$ (respectively $\A_\R$) does not depend
on the choice of $J$ (respectively of $I$) since $\widetilde{G}_\R$
(respectively $\widetilde{G}_\L$) acts transitively on the family
of intervals.
\end{definition}
If the net $\A$ is conformal, then the components $\A_\L$ and $\A_\R$ are
nontrivial (see Remark \ref{nontrivial-components})

\begin{lemma}[\cite{Rehren}]
The nets $\A_\L, \A_\R$ extend to M\"obius nets on $S^1$.
For a fixed double cone $I\times J$, there holds
\[
\A_\L(I)\vee\A_\R(J) \simeq \A_\L(I)\otimes\A_\R(J).
\]
\end{lemma}

Then we determine $\H^\tout = \H^\tin$ in terms of chiral components.
The key is the following lemma.
\begin{lemma}[\cite{Rehren}, Lemma 2.3]\label{invariant-vectors}
Let $\A$ be a M\"obius covariant net.
The subspace $\overline{\A_\L(I)\Omega}$ coincides with the subspace of
$\widetilde{G}_\R$-invariant vectors.
A corresponding statement holds for $\A_\R(J)$.
\end{lemma}
\begin{remark}
The proof of this lemma requires M\"obius covariance of the net.
On the other hand, in Section \ref{proof-noninteraction},
where we utilized the fact that the representation $U$ of $\axb\times\axb$
extends to $\2dmob$, what was really needed is that $U$ decomposes into a
direct sum of copies of the four
irreducible representations in Corollary \ref{classification-axb}.
\end{remark}

\begin{theorem}\label{wave-space}
It holds that $\H^\tout = \H^\tin = \overline{\A_\L(I)\vee\A_\R(J)\Omega}$.
\end{theorem}
\begin{proof}
As we have seen in Proposition \ref{classification-psl2r}, the spaces of
invariant vectors with respect to $\widetilde{G}_\R, \widetilde{G}_\L$ and
to positive/negative lightlike translations coincide.
Lemma \ref{invariant-vectors} tells us that $\overline{\A_\L(I)\Omega} = \H_+$
and $\overline{\A_\R(J)\Omega} = \H_-$.

As elements in $\A_\L$ are fixed under the action of $\widetilde{G}_\R$,
for $x \in \A_\L(I)$ it holds that $\Phi^\tin_+(x) = x$. 
Similarly we have $\Phi^\tin_-(y) = y$ for $y \in \A_\R(J)$.
Thus we see that
\[
x\Omega \timesi y\Omega = \Phi^\tin_+(x)\Phi^\tin_-(y)\Omega = xy\Omega \in \overline{\A_\L(I)\vee\A_\R(J)\Omega}
\]
Conversely, since $\A_\L(I)$ and $\A_\R(J)$ commute, any element in $\A_\L(I)\vee\A_\R(J)$ can
be approximated strongly by linear combinations of elements of product form $xy$.
This implies the required equality of subspaces.
\end{proof}

As a simple corollary, we have another proof of noninteraction of waves
and a relation between asymptotic completeness and chirality:
\begin{corollary}\label{noninteraction-chiral}
Let $\A$ be a M\"obius covariant net.
\begin{itemize}
\item[(a)] (same as Theorem \ref{noninteraction-bf})
We have the equality $\xi_+\timesi\xi_- = \xi_+\timeso\xi_-$ for any
pair $\xi_+ \in \H_+$ and $\xi_- \in \H_-$. In particular,
such waves do not interact.
\item[(b)] $\H^{\tout}=\H^{\tin}=\H$ if and only if $\A$ coincides with its maximal chiral subnet.
\end{itemize}
\end{corollary}
\begin{proof}
Theorem \ref{wave-space} tells us that the space of collision states of waves is generated by
chiral observables $\A_\L(I)\vee\A_\R(J)$. Lemma \ref{asymptotic-field}
assures that to investigate the S-matrix it is enough to consider
observables which generate the collision states. Then, on the space of waves
$\H_0 = \overline{\A_\L(I)\vee\A_\R(J)\Omega}$ and regarding the chiral observables,
it has been shown that a chiral net is asymptotically complete ($\H^\tout = \H^\tin = \H_0$)
and the S-matrix is trivial \cite{DT10-1,DT10-2}.

If $\H_0 \neq \H$, then by the Reeh-Schlieder property, the full net $\A$ must contain
non-chiral observables, and $\A_\L\otimes\A_\R \neq \A$. If $\H_0 = \H$,
since both $\A_\L\otimes\A_\R$ and $\A$ are M\"obius covariant,
there is a conditional expectation $E_O:\A(O) \to \A_\L\otimes\A_\R(O)$ which 
preserves $\<\cdot \Omega,\Omega\>$, but $E_O$ is in fact the identity map
since $\Omega$ is cyclic for $\A_\L\otimes\A_\R(O)$ (see Theorem \ref{takesaki}).
\end{proof}

\subsection{How large is the space of collision states?}\label{how-large}
We have seen that a part $\overline{\A_\L(I)\vee\A_\R(J)\Omega}$ of the Hilbert space $\H$
can be interpreted as the space of collision states of waves and that these waves do not
interact. Then of course it is natural to investigate the particle aspects of the
orthogonal complement of this space. We do not go into the detail of this problem here,
but restrict ourselves to a few comments.

The algebra of chiral observables $\A_\L\otimes\A_\R$ is represented on the
the full Hilbert space $\H$ in a reducible way.
One can decomposes $\H$ into a direct sum of the irreducible components
with respect to $\A_\L\otimes\A_\R$:
\[
\H = \bigoplus_i \H_{\rho_i},
\]
where $\{\rho_i\}$ are irreducible representations (see \cite{LR})
of $\A_\L\otimes\A_\R$. When $\A_\L$ and $\A_\R$ are completely rational
\cite{KLM}, then the representations $\rho_i$ are tensor products
$\rho^\L_i\otimes\rho^\R_i$ of representations $\rho^\L_i$ of $\A_\L$ and
$\rho^\R_i$ of $\A_\R$.
As we consider the maximal chiral subnet introduced by Rehren,
the vacuum representations $\rho^\L_0, \rho^\R_0$ appear only once,
in the form $\rho^\L_0\otimes\rho^\R_0$ \cite[Corollary 3.5]{Rehren}.
This representation $\rho^\L_0\otimes\rho^\R_0$ is realized
on the subspace $\H_0 = \overline{\A_\L(I)\vee\A_\R(J)\Omega}$.
Theorem \ref{wave-space} says that
the waves are contained only in $\H_0$.

Hence, when $\A$ is not chiral,
the space of collision states is at most a half of the full Hilbert space,
if we simply count the number of representations which appear in the
decomposition. A conceptually more satisfactory measure is
the index of the inclusion $[\A:\A_\L\otimes\A_\R]$. The minimal value of
the index of a nontrivial inclusion is $2$, which would mean again that
waves occupy half of the available space.
This case indeed happens: Let $\A_0$ be a M\"obius covariant net on $S^1$
with $\ZZ_2$ symmetry. If we define $\A = (\A_0\otimes\A_0)^{\ZZ_2}$, 
where $\ZZ_2$ acts on $\A_0\otimes\A_0$ by the diagonal action and
$(\A_0\otimes\A_0)^{\ZZ_2}$ is the fixed point subnet of this action, then $\A$ has
$\A_0^{\ZZ_2}\otimes\A_0^{\ZZ_2}$ as the maximal chiral subnet and the
index $[\A:\A_0^{\ZZ_2}\otimes\A_0^{\ZZ_2}]$ is $2$. But in this case
it is natural to say that the orthogonal complement can be interpreted as collision
states in a bigger net $\A_0\otimes\A_0$ which do not interact.
In general, if a given net is not the fixed point, such a reinterpretation of
the orthogonal complement as waves is impossible and the index is typically
larger than $2$.  New ideas are needed to clarify this general case.

\section{Asymptotic fields given through conditional expectations}
\subsection{Characterization of noninteracting nets}\label{characterization}
In \cite{Buchholz}, in the general setting of Poincar\'e covariant nets,
Buchholz has proved that timelike commutativity implies the absence of
interaction. The purpose of this subsection is to show a strengthened converse,
namely that if a two-dimensional Poincar\'e covariant net is asymptotically complete
and noninteracting, then under natural assumptions
it is (unitarily equivalent to) a chiral M\"obius covariant net.

For this purpose, it is appropriate to extend the definition
of a net also to unbounded regions.
Let $\A$ be a Poincar\'e covariant net.
For an arbitrary open region $O$, we define $\A(O) := \bigvee_{D\subset O}\A(D)$,
where $D$ runs over all bounded regions included in $O$
(this definition coincides with the original net if $O$ is bounded).
Among important unbounded regions are {\bf wedges}.
The standard left and right wedges are defined as follows:
\begin{eqnarray*}
W_\L &:=& \{(t_0,t_1): t_0 > t_1, t_0 < -t_1\}\\
W_\R &:=& \{(t_0,t_1): t_0 < t_1, t_0 > -t_1\}
\end{eqnarray*}
The regions $W_\L$ and $W_\R$ are invariant under Lorentz boosts.
The causal complement of $W_\L$ is $W_\R$ (and vice versa). All the regions
obtained by translations of standard wedges are still
called left- and right- wedges, respectively.
Moreover, any double cone is obtained as the intersection of a left wedge and a right wedge.
Let $O'$ denote the causal complement of $O$ in $\RR^2$ (not in $\E$).
It holds that $W_\L' = W_\R$, and
if $D = (W_\R+a)\cap(W_\L+b)$ is a double cone, $a,b\in\RR^2$, then
$D' = (W_\L+a)\cup(W_\R+b)$. It is easy to see that $\Omega$ is still cyclic and
separating for $\A(W_\R)$ and $\A(W_\L)$.

Let us introduce some additional assumptions on the structure of nets.
\begin{itemize}
\item {\bf Haag duality.} If $O$ is a wedge or a double cone, then it holds that
$\A(O)' = \A(O')$.
\item {\bf Bisognano-Wichmann property.} The modular group $\Delta^{it}$ of $\A(W_\R)$
with respect to $\Omega$ is equal to $U(\Lambda(-2\pi t))$, where
$\Lambda(t) = \left(\begin{matrix} \cosh t & \sinh t\\ \sinh t & \cosh t \end{matrix}\right)$ denotes
the Lorentz boost.
\end{itemize}
Duality for wedges (namely $\A(W_\L)' = \A(W_\R)$) follows from the Bisognano-Wichmann
property (see Proposition \ref{wedge-duality}).
If $\A$ is M\"obius covariant, then the Bisognano-Wichmann property is automatic
\cite{BGL}, and Haag duality is equivalent to strong additivity \cite{Rehren}.
Apart from M\"obius nets, these properties are common even in massive interacting
models \cite{Lechner}. Furthermore, starting with $\A(W_\L)$, it is possible to
construct a net which satisfies both properties \cite{Borchers, Lechner}.
Hence we believe that
these additional assumptions are natural and throughout this section
we assume that the net $\A$ satisfies them.

Let $\A$ be a Poincar\'e covariant net satisfying the Bisognano-Wichmann property.
We start with general remarks on asymptotic fields.
Let $\N_+^\tout$ be the von Neumann algebra generated by $\Phi_+^\tout(x)$
where $x \in \A(O)$, $O \subset W_\R$ and $O$ is bounded
\footnote{From Lemma \ref{asymptotic-field-is-expectation} it is immediate that
$\Phi_+^\tout$ naturally extends to $\A(W_\R)$, but
it is convenient to define $\N_+^\tin$ with bounded regions since
we see the relation between $\Phi_+^\tin$ and $\Phi_-^\tout$ in Lemma
\ref{in-out-equality}.}.

\begin{lemma}\label{asymptotic-field-is-expectation}
The the map $\Phi_+^\tout$ which gives the asymptotic field
is a conditional expectation (cf. \ref{expectation})
from $\A(W_\R)$ onto $\N_+^\tout$
which preserves the vacuum state $\omega := \<\Omega, \cdot \Omega\>$.
\end{lemma}
\begin{proof}
By construction, $\Phi_+^\tout(x) \in \A(W_\R)$ for such $x \in \A(O), O\subset W_\R$
as above. Recall that if $g$ is a Poincar\'e transformation,
it holds that $\Ad U(g) \Phi_+^\tout(x) = \Phi_+^\tout(\Ad U(g)(x))$ (see Lemma \ref{asymptotic-field}).
Hence $\N_+^\tout$ is invariant under Lorentz boosts $\Ad U(\Lambda(-2\pi t)), t\in \RR$.
Since we assume the Bisognano-Wichmann property, $\N_+^\tout$ is invariant
under the modular group of $\A(W_\R)$ with respect to $\omega$.

By Takesaki's Theorem \ref{takesaki}, there is a conditional expectation
$E$ from $\A(W_\R)$ onto $\N_+^\tout$ and this is implemented by
the projection $P_+^\tout$ onto $\overline{\N_+^\tout\Omega}$.
By Lemma \ref{asymptotic-field}, we know that $P_+^\tout = P_+$.
Two operators $E(x)$ and $\Phi_+^\tout(x)$ in $\A(W_\R)$ satisfy
$E(x)\Omega = P_+^\tout x\Omega = P_+ x\Omega = \Phi_+^\tout(x)\Omega$.
The vacuum vector $\Omega$ is separating for $\A(W_\R)$, hence
they coincide.
\end{proof}
Analogously, we consider $\N_-^\tin$ generated by
$\{\Phi_-^\tin(x): x \in \A(O), O \subset W_\R, O \mbox{ bounded}\}$. The map $\Phi_-^\tin$
is the conditional expectation from $\A(W_\R)$ onto $\N_-^\tin$.
\begin{proposition}\label{wedge-recovery}
Let us assume that $\A$ is asymptotically complete.
The wedge algebra $\A(W_\R)$ is generated by $\N_+^\tout$ and $\N_-^\tin$.
\end{proposition}
\begin{proof}
As we observed before Lemma \ref{asymptotic-field-is-expectation},
$\N_+^\tout$ and $\N_-^\tin$ are invariant under Lorentz boosts.
Hence the same holds for $\N_\R := \N_+^\tout \vee \N_-^\tin$.
Again by Theorem \ref{takesaki}, there is a conditional expectation $E$ from $\A(W_\R)$
onto $\N_\R$.
The wedge algebra $\A(W_\R)$ is already in the GNS representation of the vacuum $\omega$
since $\Omega$ is cyclic and separating for $\A(W_\R)$.
$\N_R\Omega$ contains all the collision states, since $\N_\R\Omega \supset \{\Phi_+^\tout(x)\Phi_-^\tin(y)\Omega\}$
and the assumption of asymptotic completeness tells us that $\N_\R\Omega$ is dense in $\H$,
hence the projection $P_{\N_\R}$ onto $\overline{\N_\R\Omega}$ is equal to $\1$.
Therefore the conditional expectation $E$ is in fact the identity map and
$\N_\R = \A(W_\R)$.
\end{proof}

\begin{lemma}\label{in-out-equality}
Let us assume that $\A$ is asymptotically complete and noninteracting.
Then it holds that $\Phi^\tout_+(x) = \Phi^\tin_+(x)$ and
$\Phi^\tin_-(x) = \Phi^\tout_-(x)$
for $x \in \A(O)$.
\end{lemma}
\begin{proof}
We present the proof  for ``$+$'' objects only, since the other assertion is analogous.
By the assumption that $S=\1$, it follows that
$\xi_+ \timesi \xi_- = \xi_+ \timeso\xi_-$ for any pair
$\xi_+ \in \H_+, \xi_- \in \H_-$. Then we have
\begin{eqnarray*}
\Phi^\tout_+(x) \cdot\xi_+\timeso\xi_-
&=& (\Phi^\tout_+(x)\xi_+)\timeso\xi_- \\
&=& P_+ x\xi_+\timeso \xi_- \\
&=& P_+ x\xi_+\timesi \xi_- \\
&=& (\Phi_+^\tin(x)\xi_+)\timesi\xi_- \\
&=& \Phi^\tin_+(x)\cdot \xi_+\timesi\xi_- \\
&=& \Phi^\tin_+(x)\cdot \xi_+\timeso\xi_-,
\end{eqnarray*}
where, in the 1st and 5th lines we used the fact that right- and left- moving
asymptotic fields commute, the 2nd and 4th equalities come from Lemma \ref{asymptotic-field}
and the rest is particular cases of the equivalence between ``$\timesi$'' and ``$\timeso$''.
By the assumption of asymptotic completeness,
$\xi_+\timesi\xi_-=\xi_+\timeso\xi_-$ span the whole space,
hence we have the equality of operators $\Phi^\tout_+(x) = \Phi^\tin_+(x)$.
\end{proof}

\begin{lemma}\label{wedge-split}
Let us assume that $\A$ is asymptotically complete and noninteracting.
The map
\[
W:\xi_+\otimes \xi_- \mapsto \xi_+\timesi\xi_- = \xi_+\timeso\xi_-
\]
gives a natural unitary equivalence $(P_+\N_+^\tout)\otimes(P_-\N_-^\tin) \simeq \A(W_\R)$,
which is elementwise expressed as
$P_+\Phi_+^\tout(x) \otimes P_-\Phi_-^\tin(y) \mapsto \Phi_+^\tout(x)\Phi_-^\tin(y)$.
Furthermore, this decomposition is compatible with
the action of the Poincar\'e group $\poincare$:
$\H_+$ and $\H_-$ are invariant under $\poincare$, hence there is a tensor product
representation on $\H_+\otimes\H_-$ and it holds that
$W\cdot(U(g)P_+\Phi_+^\tout(x)\Omega\otimes U(g)P_-\Phi_-^\tin(y)\Omega)
= U(g)W\cdot (P_+\Phi_+^\tout(x)\Omega\otimes P_-\Phi_-^\tin(y)\Omega$).
\end{lemma}
\begin{proof}
The unitarity of the map $W$ in the statement is clear from Lemma \ref{collision-states}
and it follows that $W$ intertwines the actions of asymptotic fields by
Lemma \ref{asymptotic-field}: Namely, $\Phi_+^\tout$ and $\Phi_-^\tout$ act as
in a tensor product (Lemma \ref{asymptotic-field}, \ref{collision-states})
but we know that $\Phi_-^\tout(x) = \Phi_-^\tin(x)$ from noninteraction (Lemma \ref{in-out-equality}).
As for the action of the Poincar\'e group,
we see from Lemma \ref{asymptotic-field},
for $x$ and $y$  as in Lemma \ref{in-out-equality}, that
\begin{eqnarray*}
W\cdot U(g)\Phi_+^\tout(x)\Phi_-^\tin(y)\Omega
&=& W\cdot\Ad U(g)(\Phi_+^\tout(x))\Ad U(g)(\Phi_-^\tin(y))\Omega \\
&=& W\cdot\Phi_+^\tout(\Ad U(g)(x))\Phi_-^\tin(\Ad U(g)(y))\Omega \\
&=& P_+\Phi_+^\tout(\Ad U(g)(x))\Omega\otimes P_-\Phi_-^\tin(\Ad U(g)(y))\Omega \\
&=& P_+ U(g)\Phi_+^\tout(x)\Omega\otimes P_- U(g)\Phi_-^\tin(y)\Omega \\
&=& U(g)P_+\Phi_+^\tout(x)\Omega\otimes U(g)P_-\Phi_-^\tin(y)\Omega,
\end{eqnarray*}
where in the last step we used the fact that $\H_+$ and $\H_-$ are invariant
under $U(g)$. This completes the proof.
\end{proof}

Let us recall the notion of a half-sided modular inclusion
due to Wiesbrock, with which we recover the M\"obius symmetry of a given noninteracting net.
\begin{theorem}[\cite{Wiesbrock, AZ}]
Let $\N \subset \M$ be an inclusion of von Neumann algebras, $\Omega$
be a cyclic and separating vector for $\N,\M$ and $\M\cap\N'$.
Let us assume that the modular group $\s^\Omega_t$ of $\M$ with respect to the state
$\<\Omega,\cdot \Omega\>$ preserves $\N$ with $t\ge 0$ (respectively $t\le 0$).
Then there is a M\"obius covariant net $\A_0$ on $S^1$ such that $\A_0(\RR_-) = \M$
and $\A_0(\RR_- - 1) = \N$ (respectively $\A_0(\RR_+) = \M$ and $\A_0(\RR_+ + 1) = \N$).

If a unitary representation $T_0$ of $\RR$ with positive spectrum satisfies
$T_0(t)\Omega = \Omega$ for $t\in\RR$,
$\Ad T_0(t)(\M) \subset \M$ for $t\le 0$ (respectively $t\ge 0$) and $\Ad T_0(-1)(\M) = \N$
(respectively $\Ad T_0(1)(\M) = \N$), then $T_0$ is the representation of the translation
group of the M\"obius covariant net constructed above.
\end{theorem}
Such an inclusion $\N \subset \M$ is called a {\bf standard $\pm$half-sided modular inclusion}
(standardness refers to the condition that $\Omega$ is cyclic and separating for
$\M\cap\N'$).
If $T_0$ is the representation of the translation group, the M\"obius
net on $S^1$ restricted to the real line $\RR$ has an explicit form \cite{Wiesbrock,AZ}
\begin{align*}
\A_0((s,t)) &= T_0(s)\M' T_0(s)^* \cap T_0(t)\M T_0(t)^*\\
(\mbox{respectively }\A_0((s,t)) &= T_0(s)\M T_0(s)^*\cap T_0(t)\M' T_0(t)^*).
\end{align*}

For a von Neumann algebra $\N$ on the Hilbert space $\H$ (on which the net $\A$
is defined), we denote
$\N(a) = \Ad T(a)(\N)$ for $a\in\RR^2$, where $T$ is the representation of
the translation group for the net $\A$ (see Section \ref{scatteringtheory}). 
We put $a_{1} := (1,1), a_{-1} := (-1,1) \in \RR^2$.
\begin{lemma}\label{hsmi}
The inclusion $P_+\N_+^\tout(a_{-1}) \subset P_+\N_+^\tout$ is a standard
$+$half-sided modular inclusion with respect to $\Omega$ on $\H_+$.
Analogously, $P_-\N_-^\tin(a_{1}) \subset P_-\N_-^\tin$ is a standard
$-$half-sided modular inclusion with respect to $\Omega$ on $\H_-$.
\end{lemma}
\begin{proof}
We prove only the former claim, since the latter is analogous.
Recall that the conditional expectation $\Phi_+^\tout$ commutes with
translations (Lemma \ref{asymptotic-field}), hence $\N_+^\tout(a_{-1})$ is generated by
$\{\Phi_+^\tout(x): x \in \A(O), O\subset W_\R + a_{-1}, O \mbox{ bounded}\}$.
The region $W_\R + a_{-1}$ is mapped into itself by Lorentz boosts $\Lambda(-t), t\ge 0$.
Lemma \ref{asymptotic-field-is-expectation} tells us that $\Phi_+^\tout$ is a conditional
expectation which preserves $\omega := \<\Omega,\cdot \Omega\>$,
hence the modular automorphism of $\N_+^\tout$
with respect to $\omega$ is the restriction of the modular automorphism of $\A(W_\R)$.
Thus Bisognano-Wichmann property shows that $\N_+^\tout(a_{-1})$ is
invariant under the modular automorphism $\s^\Omega_t$ of $\N_+^\tout$ for $t\ge 0$.
The projection $P_+$ commutes with both of $\N_+^\tout$ and $\N_+^\tout(a_{-1})$, hence
it is a $+$half-sided modular inclusion.

As for standardness, note that $\A(W_\R)\cap\A(W_\L+a_{-1}+a_1)$ contains
$\A(D)$ where $D = W_\R \cap (W_\L+a_{-1}+a_1)$ is a double cone.
Recall that $\A(W_\R) \simeq P_+\N_+^\tout\otimes P_-\N_-^\tin$ and
the action of the Poincar\'e group splits as well (Lemma \ref{wedge-split}).
According to this unitary equivalence we have
$\A(W_\L+a_{-1}+a_1) \simeq P_+\N_+^\tout(a_{-1})'\otimes P_-\N_-^\tin(a_{1})'$ and
$\A(W_\R)\cap\A(W_\L+a_{-1}+a_1)
\simeq P_+(\N_+^\tout\cap\N_+^\tout(a_{-1})')\otimes P_-(\N_-^\tin\cap\N_-^\tin(a_{1})')$,
since we have wedge duality (Proposition \ref{wedge-duality}).
The vector $\Omega \simeq \Omega\otimes\Omega$ is cyclic for $\A(D)$ (Reeh-Schlieder property)
and this is possible only if $\Omega$ is cyclic for both $P_+\left(\N_+^\tout\cap\N_+^\tout(a_{-1})'\right)$
and $P_-\left(\N_-^\tin\cap\N_-^\tin(a_{1})'\right)$.
The cyclicity for $P_+\left(\N_+^\tout\cap\N_+^\tout(a_{-1})'\right)$ is the standardness.
\end{proof}

\begin{theorem}\label{noninteracting-chiral}
Let $\A$ be a Poincar\'e covariant net, asymptotically complete and noninteracting
(satisfying Haag duality and Bisognano-Wichmann property).
Then it is a chiral M\"obius covariant net.
\end{theorem}
\begin{proof}
First we have to prepare two M\"obius covariant nets on $S^1$:
This has been done in Lemma \ref{hsmi}. Namely, putting $a_{\pm t} = (t,\pm t) \in \RR^2$ for
$t\in\RR$, we have two nets
\begin{align*}
\A_\L((s,t)) &= P_+\left(\N_+^\tout(a_{-s})' \cap \N_+^\tout(a_{-t})\right), \\
\A_\R((s,t)) &= P_-\left(\N_-^\tin(a_{s}) \cap \N_-^\tin(a_{t})'\right).
\end{align*}
Under the unitary equivalence between $\H$ and $\H_+\otimes\H_-$ from
Lemma \ref{wedge-split}, Haag duality implies that, for the double cone
$D = W_\R \cap (W_\L+a_{-1}+a_1)$, we have
\begin{eqnarray*}
\A(D)
&=& \A(W_\R)\cap\A(W_\L+a_{-1}+a_1) \\
&\simeq& P_+(\N_+^\tout\cap\N_+^\tout(a_{-1})')\otimes P_-(\N_-^\tin\cap\N_-^\tin(a_{1})') \\
&=& \A_\L((-1,0))\otimes\A_\R((0,1)).
\end{eqnarray*}
The corresponding equality for general intervals $(s_\L,t_\L),(s_\R,t_\R)$ follows
 from the above definition of nets $\A_\L,\A_\R$. In Lemma \ref{wedge-split} we saw
that the actions of the Poincar\'e group are compatible with this unitary
equivalence.
\end{proof}

\begin{remark}
Haag duality is used only in Theorem \ref{noninteracting-chiral}.
Since a Poincar\'e covariant net $\A$ with Bisognano-Wichmann property is
wedge dual (Propositions \ref{wedge-duality}, \ref{essential-duality}),
we can infer that the dual net $\A^\mathrm{d}$ (see \cite{Baumgaertel})
is a chiral M\"obius net even if we do not assume Haag duality.
\end{remark}

Combining this and Corollary \ref{ac-dilation}, we see the following:
\begin{corollary}\label{extension-dilation}
An asymptotically complete, Poincar\'e-dilation covariant net $\A$
(satisfying Haag duality and Bisognano-Wichmann property) is
(unitarily equivalent to) a chiral M\"obius covariant net.
\end{corollary}

\subsection{Asymptotic fields in M\"obius covariant nets}\label{af-in-mob}
Finally, as a further consequence of the considerations on conditional expectations,
we show that in- and out-asymptotic fields coincide
in M\"obius covariant nets even without assuming the asymptotic completeness.
Lemma \ref{asymptotic-field-is-expectation} has been proved for general
Poincar\'e covariant nets with Bisognano-Wichmann property,
hence it applies to M\"obius covariant nets as well (see (10M) in Section \ref{conformalnets}).
We use the same notations as in Section \ref{subspace}.

Let $\A_\L\otimes\A_\R$ be the maximal chiral subnet. Since both nets
$\A$ and $\A_\L\otimes\A_\R$ are M\"obius covariant, they satisfy
Bisognano-Wichmann property in $\E$ ((10M) in Section \ref{conformalnets}).
Theorem \ref{takesaki} of Takesaki implies that there is a conditional
expectation $E_D$ from $\A(D)$ onto $\A_\L(I)\otimes\A_\R(J)$, where
$D = I\times J$ is a double cone in $\E$, which is implemented by the projection $P$
onto $\H^\tin = \H^\tout = \overline{\A_\L(I)\vee\A_\R(J)\Omega}$
(see Theorem \ref{wave-space}). Since the projection $P$ does not depend on
$D$, the conditional expectations $\{E_D\}_{D\subset\E}$ are compatible,
namely, if $D_1\subset D_2$ then it holds that $E_{D_2}|_{\A(D_1)} = E_{D_1}$.
Indeed, it holds that $E_{D_1}(x)\Omega = Px\Omega = E_{D_2}(x)\Omega$ and
$\Omega$ is separating for $\A(D_2)$.

In addition, there is a conditional expectation $\id\otimes\omega$
from $\A_\L(I)\otimes\A_\R(J)\simeq\A_\L(I)\vee\A_\R(J)$
onto $\A_\L(I)$ which obviously preserves $\omega$ and is
implemented by $P_+$ (see Theorem \ref{takesaki}). If we take intervals
$I_1 \subset I_2$, then the corresponding expectations are obviously
compatible. By composing this expectation
and $E_D$, we find an expectation $E_\L$ from $\A(D)$ onto $\A_\L(I)$ which preserves
$\omega$ and is implemented by $P_+$ (we omit the dependence on $D$ since this
family of expectations is compatible).
Analogous statements hold for $\A_\R(J)$.
\begin{theorem}
If $\A$ is a M\"obius covariant net, then for $x \in \A(D)$ with some
bounded double cone $D = I\times J$,
it holds that $\Phi_+^\tout(x)=\Phi_+^\tin(x)$ and
$\Phi_-^\tout(x)=\Phi_-^\tin(x)$.
\end{theorem}
\begin{proof}
We exhibit the proof only for ``$+$'' objects since the other is analogous.
As we have seen in Lemma \ref{asymptotic-field-is-expectation},
$\Phi_+^\tout$ is a conditional expectation from $\A(W_\R)$ onto $\N_+^\tout$
which preserves $\omega$.

We claim that $\Phi_+^\tout(x) = E_\L(x)$.
We may assume that $D \subset W_\R$ since $\Phi_+^\tout$ commutes
with translations, and $E_\L$ is compatible and the translated expectation
$\Ad T(a)\circ E_\L \circ \Ad T(-a)$ still preserves $\omega$
(hence $E_\L$ commutes with translation $\Ad T(a)$ as well).
It holds that $\Phi_+^\tout(x) \subset \A(W_\R)$ and
$E_\L(x) \in \A_\L(I) \subset \A(D) \subset \A(W_\R)$.
In addition we have
$\Phi_+^\tout(x)\Omega = P_+ x\Omega = E_\L(x)\Omega$,
hence by the separating property of $\Omega$ for $\A(W_\R)$ we obtain
the claimed equality.

Similarly one sees $\Phi_+^\tin(x) = E_\L(x)$, hence two
maps $\Phi_+^\tout$ and $\Phi_+^\tin$ coincide.
\end{proof}

\section{Concluding remarks}\label{concluding}
In the first part of this work we showed that waves in two-dimensional M\"obius nets
do not interact.
This result can be seen as a (non-trivial) adaptation of an argument of Buchholz and Fredenhagen \cite{BF}
to the two-dimensional case. Moreover, we showed
that collision states of waves correspond precisely to
states generated from the vacuum by observables from the maximal chiral subnet.
This implies the equivalence between asymptotic completeness and chirality of a given
M\"obius covariant theory. As we observed in \ref{chiral-components}, chiral observables
admit geometric definitions. This is a special feature of two-dimensional M\"obius theory,
which, to our knowledge, does not have a counterpart in higher-dimensional theories.

The second part of this paper relies on our observation that, in a Poincar\'e covariant net 
with the Bisognano-Wichmann property, the maps which give asymptotic fields are conditional expectations.
Exploiting this fact, we showed that a Haag dual net is asymptotically complete and noninteracting if and only if
it is a chiral M\"obius net.
We also strengthened our result on noninteraction by showing that in- and out-asymptotic 
fields in any (possibly non-chiral) M\"obius net coincide.

The orthogonal complement of the space of collision states,
which may be quite large as we explained in Section \ref{how-large},
is a natural subject of future research.
Fortunately, we have tools to investigate
this orthogonal complement: They include the theory of
particle weights \cite{BPS, Porrmann1}, developed to study infraparticles.
With the help of this theory we have confirmed that infraparticles are present in all states
in product representations of the chiral subnet, hence in the orthogonal complement
of the space of collision states of waves in any completely rational net
\cite{DT10-1,KLM}.
The question of interaction and asymptotic completeness of these excitations remains
open to date (for a general account on asymptotic completeness, see \cite{Buchholz87}).
However, the fact that the incoming and outgoing asymptotic fields coincide in
M\"obius covariant theories on the entire Hilbert space suggests the absence of interaction.
These issues are under investigation \cite{DT11-3}.

\subsubsection*{Acknowledgment.}
I wish to thank Wojciech Dybalski for his many advices and
detailed reading of the manuscript, Roberto Longo for his constant support
and Yasuyuki Kawahigashi for his
useful suggestion.
A part of this work has been done during my stay at University of G\"ottingen in
August 2010. The hospitality of the institute and of the members of Mathematical Physics
group, in particular Karl-Henning Rehren and Daniela Cadamuro, is gratefully acknowledged.

\appendix
\newcommand{\appsection}[1]{\let\oldthesection\thesection
  \renewcommand{\thesection}{Appendix \oldthesection}
  \section{#1}\let\thesection\oldthesection}

\appsection{Remarks on conditional expectations}\label{expectation}
The Bisognano-Wichmann property asserts a relation between the dynamics
of the net and the Tomita-Takesaki modular theory. In the modular theory, one of the fundamental
tools is conditional expectation. We briefly recall here its definition
and discuss some immediate consequences.
A {\bf conditional expectation} from a von Neumann algebra $\M$ onto
a subalgebra $\N$ is a linear map $E: \M \to \N$ satisfying the following
properties:
\begin{itemize}
\item $E(x) = x$ for $x \in \N$.
\item $E(xyz) = xE(y)z$ for $x,z \in \N$, $y \in \M$.
\item $E(x)^*E(x) \le E(x^*x)$ for $x \in \M$.
\end{itemize}
We see in Section \ref{characterization} that the maps which give asymptotic fields
can be considered as conditional expectations
between appropriate von Neumann algebras. Let us recall the fundamental
theorem of Takesaki \cite[Theorem IX.4.2]{Takesaki}.
\begin{theorem}\label{takesaki}
Let $\N \subset \M$ be an inclusion of von Neumann algebras and $\f$ be a
faithful normal state on $\M$. Then the following are equivalent:
\begin{itemize}
\item $\N$ is invariant under the modular automorphism group $\s^\f_t$.
\item There is a normal conditional expectation $E$ from $\M$ onto $\N$ such that
$\f = \f \circ E$.
\end{itemize}
Furthermore, if the above conditions hold, then the conditional expectation
$E$ is implemented by a projection in the following sense: We consider
the GNS representation $\pi_\f$ and $\Phi$ be the GNS vector. Let $P_\N$
be the projection onto the subspace $\overline{\N\Phi}$. Then it holds
that $E(x)\Phi = P_\N x\Phi$. In particular, $\N = \M$ if and only if
$P_\N = \1$ (hence E = \id).
The modular automorphism group of $\f|_\N$ is equal to $\s^\f|_\N$.
\end{theorem}

A Poincar\'e covariant net $\A$ is said to be {\bf wedge dual} if
it holds that $\A(W_\L)' = \A(W_\L') (= \A(W_\R))$ (see Section \ref{characterization}
for $W_\L$ and $W_\R$).
With the help of conditional expectation, it is easy to deduce that
Bisognano-Wichmann property (see Section \ref{characterization}) implies wedge duality,
although this implication has been essentially known \cite{Borchers, Rigotti}.

\begin{proposition}\label{wedge-duality}
If a Poincar\'e covariant net $\A$ satisfies Bisognano-Wichmann property,
then it is wedge dual.
\end{proposition}
\begin{proof}
The modular automorphism group $\s^\Omega_t$ of $\A(W_\L)'$ is implemented
by $\Delta_\Omega^{-it}$, which is equal to $U(\Lambda(-2\pi t))$ by
Bisognano-Wichmann property. It is obvious that $\A(W_\R) \subset \A(W_\L)'$
and $\A(W_\R)$ is invariant under $\Ad U(\Lambda(-2\pi t)) = \Ad \Delta_\Omega^{-it}$.
Hence by Takesaki's Theorem \ref{takesaki}, there is a conditional expectation
$E$ from $\A(W_\L)'$ onto $\A(W_\R)$ which preserves $\omega = \<\Omega, \cdot \Omega\>$
and it is implemented by the projection onto the subspace $\overline{\A(W_\R)\Omega}$.
But by Reeh-Schlieder property it is the whole space $\H$, hence $E$ is the
identity map and we obtain $\A(W_\L)' = \A(W_\R)$.
\end{proof}

For a given net $\A$, we can associate the {\bf dual net} $\A^\mathrm{d}$
(\cite[Section 1.14]{Baumgaertel}), defined by
\[
\A^\mathrm{d}(O_0) = \bigcap_{O\perp O_0} \A(O)',
\]
where $O \perp O_0$ means that $O$ and $O_0$ are causally disjoint.
$\A^\mathrm{d}$ does not necessarily satisfy locality nor additivity.
Additivity is usually necessary only in proving Reeh-Schlieder property,
so we do not discuss here. We have the following.

\begin{proposition}[\cite{Baumgaertel}]\label{essential-duality}
If a Poincar\'e covariant net $\A$ is wedge dual,
then $\A^\mathrm{d}$ is local and Haag dual.
\end{proposition}
Thus, if we consider the dual net $\A^\mathrm{d}$ as a natural extension,
Haag duality for a net with Bisognano-Wichmann property is not
a strong requirement and the only essential additional assumption in
Section \ref{characterization} is Bisognano-Wichmann property.

\appsection{Chiral components of conformal nets}\label{chiral-components}
In this appendix we consider various definitions of chiral components
when a net $\A$ is conformal.
These observations are not needed in the main text at the technical level but
justify the interpretation of chiral observables as observables localized on lightlines.

We use the notations from Section \ref{subspace}.
\begin{proposition}\label{component1}
For distant intervals $J_1, J_2 \subset \RR$ (i.e.\! they are disjoint with
nonzero distance), it holds that
\[
\A_\L(I) = \A(I\times J_1)\cap\A(I\times J_2).
\]
\end{proposition}
\begin{proof}
Since the M\"obius group $\psl2r = G_\R$ acts transitively on the
family of intervals in $\RR \subset S^1$, the inclusion
$\A_\L(I) \subset \A(I\times J)$ holds for any interval $J$ by
the covariance of the net $\A$. Thus inclusions in one direction is proven.

To see the converse inclusion, we consider the commutants.
By the Haag duality in $\E$, we have
\[
\A_\L(I)' = (\A(I\times J) \cap U(\widetilde{G}_\R)')' = \A(I^+\times J^-) \vee U(\widetilde{G}_\R)
\left( = \A(I^-\times J^+)\vee U(\widetilde{G}_\R)\right),
\]
where $I^{\pm},J^{\pm}$ are defined in Section \ref{conformalnets}, and
\[
\left(\A(I\times J_1)\cap\A(I\times J_2)\right)' = \A(I^+\times J_1^-)\vee \A(I^+\times J_2^-).
\]

Recall that we can choose an arbitrary $J$. Let $J$ be an interval which
includes both $J_1$ and $J_2$. In this case we have $J^- \subset J_1^-$ and
$J^- \subset J_2^-$, hence
\[
\A(I^+\times J^-) \subset \A(I^+\times J_1^-)\vee\A(I^+\times J_2^-).
\]
Furthermore, the fact that $J_1$ and $J_2$ are distant implies that
the family of (two) intervals $\{J_1^-,J_2^-\}$ is an open cover of a closed interval of
length $2\pi$. The algebra $\A(I^+\times J_1^-)$ (respectively $\A(I^+\times J_2^-)$)
contains the representatives of diffeomorphisms of the form $\id \times g_\R$ with
$\supp(g_\R) \subset J_1$ (respectively $\supp(g_\R) \subset J_2$) in the sense that
$\conf$ is a quotient group of $\overline{\diffs1}\times\overline{\diffs1}$
(see Section \ref{conformalnets}).

We claim that the algebra $\A(I^+\times J_1^-)\vee\A(I^+\times J_2^-)$ contains
any representative of the form $\id \times g_\R$ where
$g_-$ is an arbitrary element in $\overline{\diffs1}$.
Note that $\overline{\diffs1}$ can be identified with the group of diffeomorphisms of $\RR$
commuting with the translation by $2\pi$ and an element of the form $\id\times g_\R$
where $\supp(g_\R) \subset J_1^-$ or $\supp(g_\R) \subset J_2^-$ can be viewed as a diffeomorphism
with a periodic support. The group $\overline{\diffs1}$ is generated by such elements,
hence we obtain the claim. In particular it contains
the representatives of the universal cover $\widetilde{G}_\R$ of the M\"obius group.
Summing up, we have shown the inclusion
\[
\A(I^+\times J^-)\vee U(\widetilde{G}_\R) \subset \A(I^+\times J_1^-)\vee\A(I^+\times J_2^-).
\]
The commutant of this relation gives the required inclusion.
\end{proof}
In \cite{KL03}, the intersection $\bigcap_J \A(I\times J)$ is taken as the definition of
the chiral component. In fact, this and Definition \ref{chiral-component} coincide under
the diffeomorphism covariance.
\begin{corollary}
We have $\A_\L(I) = \bigcap_{J} \A(I\times J)$. Here, the intersection
can be taken over the set of finite length intervals contained in $\RR = S^1\setminus \{-1\}$ or
even all intervals in the universal covering space of $S^1$ by considering
$\A$ as a net on $\E$.
\end{corollary}
\begin{remark}\label{nontrivial-components}
 From Proposition \ref{chiral-net-contains-virasoro}, it follows that,
for a conformal net $\A$, $\A_\L(I)$ contains the representatives
of diffeomorphisms $g_\L\times \id$ with $g_\L$ supported in $I$
and hence it is nontrivial, although the intersection of regions $\bigcap_{J} I\times J$ is empty.
A similar statement holds for $\A_\R$.
\end{remark}

If the chiral components $\A_\L,\A_\R$ satisfy strong additivity,
another (potentially useful) definition is possible. This should support an
intuitive view that $\A_\L,\A_\R$ live on lightrays.
\begin{proposition}\label{component2}
Assume that $\A_\R$ is strongly additive.
If $\{J_n\}$ is a sequence of intervals shrinking to a point,
then it holds that $\A_\L(I) = \bigcap_n\A(I\times J_n)$.
\end{proposition}
\begin{proof}
First we claim that $\A_\L(I) = \A(I\times J_1)\cap\A(I\times J_2)$
if $J_1$ and $J_2$ are obtained from an interval $J$ by removing
an interior point. One sees that the proof of Proposition \ref{component1}
works except the part concerning the diffeomorphisms.
Namely, it holds that $\A_\R(J_1\cup J_2) \subset \A(I^+\times J_1^-)\vee\A(I^+\times J_2^-)$

This time, the union $J_1^-\cup J_2^-$ is of length $2\pi$. 
By the assumed strong additivity of the component $\A_\R$,
this implies that $\A_\R(S^1) \subset \A(I^+\times J_1^-)\vee\A(I^+\times J_2^-)$.
In fact, if $J$ is an interval with length less than $2\pi$
which contains a boundary point of $J_1^-\cup J_2^-$, then $\A_\R(J)$ is
contained in $\A(I^+\times J_1^-)\vee\A(I^+\times J_2^-)$ by strong additivity
(note that the restriction of $\A_\R(I)$ to its vacuum representation is injective
if $I$ is a bounded interval). By the additivity of the chiral component,
$\A(I^+\times J_1^-)\vee\A(I^+\times J_2^-)$ contains all the representatives of
diffeomorphisms of the form $\id \times g$, $g\in\overline{\diffs1}$, in particular
representatives of $\id\times \widetilde{G}_\R$. The rest of the argument
is the same as Proposition \ref{component1}.

Let $\{J_n\}$ be a sequence of intervals shrinking to a point,
where $J_n = (a_n,b_n)$.
Let $\{J_{1,n}\}$ and $\{J_{2,n}\}$ be sequences of intervals
which are obtained by $J_{1,n} = (a_0,b_n)$ and $J_{2,n} = (a_n,b_0)$.
Let us denote $J_1 := \interior(\bigcap_n J_{1,n}) = (a_0,c), J_2 := \interior(\bigcap_n J_{2,n}) =(c,b_0)$,
where $c = \lim_n a_n = \lim_n b_n$ and $\interior(\cdot)$ means the interior.
It is clear that
\[
\A_\L(I) \subset \A(I\times J_n) \subset \A(I\times J_{1,n})\cap \A(I\times J_{2,n}),
\]
but the last expression tends to
\begin{eqnarray*}
\bigcap_n \A(I\times J_{1,n})\cap \A(I\times J_{2,n})
&=& \left(\bigcap_n \A(I\times J_{1,n})\right) \cap \left(\bigcap_n \A(I\times J_{2,n})\right) \\
&=& \A(I\times J_1)\cap \A(I\times J_2),
\end{eqnarray*}
where the last equality follows from the Haag duality in $\E$ and additivity.
We have seen that this is equal to $\A_\L(I)$,
hence the intersection $\bigcap_n \A(I\times J_n)$ is equal to this as well.
\end{proof}
\begin{remark}
Rehren defined the ``generating property'' of the net by
\begin{align*}
&U(\widetilde{G}_\L) \subset \A(I\times J) \vee \A(I'\times J)\\
&U(\widetilde{G}_\R) \subset \A(I\times J) \vee \A(I\times J'), 
\end{align*}
for any $I,J$.
We proved Proposition \ref{component2} by showing the generating
property for $\A$ with the strongly additive conformal components.
It has been shown in \cite{Rehren} that the generating property implies that
$\A_\L(I) = \A(I\times J_1)\cap\A(I\times J_2)$ where $J_1$ and $J_2$ are
obtained by removing an interior point from an interval.
\end{remark}

\end{document}